%% file: Fractionalchrom_v6.tex
\documentclass[11pt,onecolumn]{IEEEtran}

\usepackage{mathpazo} 

\usepackage{cite}
\usepackage{graphicx,color,epsfig}
\usepackage{amsfonts,amsmath,amssymb}

\input{macros}

\newtheorem{definition}{Definition}
\newtheorem{theorem}{Theorem}
\newtheorem{lemma}{Lemma}
\newtheorem{corollary}{Corollary}

\begin{document}

\title{Local Graph Coloring and Index Coding}
\author{
\IEEEauthorblockN{Karthikeyan Shanmugam \IEEEauthorrefmark{1}, Alexandros G. Dimakis \IEEEauthorrefmark{1} and Michael Langberg \IEEEauthorrefmark{2}}
 
\IEEEauthorblockA{\IEEEauthorrefmark{1}
Department of Electrical and Computer Engineering \\
University of Texas at Austin \\
Austin, TX 78712-1684 \\
karthiksh@utexas.edu,dimakis@austin.utexas.edu}

\IEEEauthorblockA{\IEEEauthorrefmark{2}
Department of Mathematics and Computer Science \\
The Open University of Israel \\
Raanana 43107, Israel \\
\{mikel\}@open.ac.il}
}
\date{\today}

\maketitle

\begin{abstract}
	We present a novel upper bound for the optimal index coding rate. Our bound uses 
	a graph theoretic quantity called the local chromatic number. We show how a good 
	local coloring can be used to create a good index code. The local coloring 
	is used as an alignment guide to assign index coding vectors from a general position MDS code. 
	We further show that a natural
	LP relaxation yields an even stronger index code. 
	Our bounds provably outperform the state of the art 
	on index coding but at most by a constant factor. 
\end{abstract}	 
	

\section{Introduction}

Index coding is a multiuser communication problem that is recently receiving significant attention. 
It is perhaps the simplest possible network problem (since it can mapped 
on a communication network with only one finite capacity link) and yet remarkable connections to many 
information theory problems have been recently discovered. The problem was introduced by Birk and Kol~\cite{birk1998informed} and has received significant attention, see \textit{e.g.} 
\cite{bar2006index,alon2008broadcasting,blasiak2010index,maleki2012index}.
It was very recently shown that any (even nonlinear) network coding problem on an arbitrary graph 
can be reduced to an equivalent index coding instance~\cite{effros2012equivalence}. 
Further, intriguing connections between index coding and the concept of interference alignment appear in~\cite{maleki2012index}.
Our work addresses the construction of index codes.
Conceptually, we view index coding as an interference alignment problem, and show how index codes can be created by exploiting a special type of graph coloring. 

In 1986, Erd\H{o}s \textit{et al.}~\cite{erdos1986coloring} defined the \textit{local chromatic number} of a graph. Given an undirected graph $G$, a coloring of the vertices is called proper if no two adjacent vertices receive identical colors. The \textit{chromatic number} $\chi(G)$ is the minimum number of colors needed for such a proper coloring. The local chromatic number $\chi_{\ell}(G)$ is defined as the maximum number of different colors that appear in the closed neighborhood of any vertex, minimized over all proper colorings. Here, a closed neighborhood of vertex $v$ includes $v$ and all its neighbored vertices.
For example, the five cycle $C_5$ has local chromatic number of $3$ obtained by {\em any} valid coloring (as the closed neighborhood of any vertex in $C_5$ is of size $3$) and it also has a chromatic number of $3$.
In general, it is clear that 
\[
\chi_{\ell}(G) \leq \chi(G).
\]  
Erd\H{o}s \textit{et al.}~\cite{erdos1986coloring} showed the non-trivial fact that the local chromatic number can indeed be arbitrarily smaller than $\chi(G)$.

Index coding is defined on a side-information graph $G$. Significant recent attention~\cite{bar2006index,maleki2012index,blasiak2010index} has been focused on obtaining bounds for the optimal communication rate $\beta(G)$. An achievable 
scheme for index coding on undirected graphs is the number of cliques required to cover $G$, which is equal to 
the chromatic number of the complement graph $\chi (\bar{G})$. This corresponds to the well-known bound
\[ 
\beta(G) \leq \chi (\bar{G}).
\]
The natural LP relaxation of this quantity is called the \textit{fractional chromatic number} $\chi_f(G)$ which also corresponds to an achievable (vector-linear) index code (as shown in \cite{blasiak2010index}). Therefore, it is known that 
\begin{equation}
	\beta(G) \leq \chi_f (\bar{G}) \leq \chi (\bar{G}),
	\label{frac_bound}
\end{equation}
and both inequalities can be strict for certain graphs. The fractional chromatic number is the best known general bound for index coding~\cite{blasiak2010index}.

\subsection{Our contribution} 

In this paper, we show that the \textit{local chromatic number} provides an achievable index coding bound:
\[
	\beta(G) \leq \chi_{\ell} (\bar{G}).
\]
For directed graphs $G_d$, a natural generalization of the local chromatic number was defined by 
K{\"o}rner \textit{et al.}~\cite{korner2005local}. This was introduced to bound the Sperner capacity of a graph 
(the natural directed generalization of the Shannon graph capacity~\cite{korner1998zero}).

The local chromatic number of a directed graph $G_d$, denoted $\chi_{\ell} (G_d)$, is the number of colors in the 
most colorful closed out-neighborhood of a vertex (defined formally later). We show that for any directed side information graph $G_d$:
\[
	\beta(G_d) \leq \chi_{\ell} (\bar{G_d}).
\]
where $\bar{G}_d$ is the directed complement of $G_d$.
We also show that the natural LP relaxation of the local chromatic number, called the 
\textit{fractional local chromatic number} $\chi_{f \ell}$ is a stronger bound on index coding:
\[
	\beta(G_d) \leq \chi_{f \ell} (\bar{G_d}).
\]
Note that there exist (directed) graphs where the fractional local chromatic number is strictly smaller than the fractional chromatic number. 
\[
\chi_{f \ell}(G_d) < \chi_{f}(G_d).
\]

\subsection{Comparison with previous results} 

We investigate the relation of our bounds to previously known results. For undirected graphs $G$
(equivalently, bi-directed digraphs), previous graph theoretic work~\cite{korner2005local} established that 
\[
\chi_{f\ell} (G)=\chi_{f} (G).
\] 
Therefore for undirected graphs we obtain no new interesting bound. 

For directed graphs, however, we show that there can be a linear additive gap between 
the local chromatic number and the fractional chromatic number. We explicitly construct a directed graph 
where $\chi_{f}=n$ and  $\chi_{\ell} \leq  \frac{n}{2} +1$.

In terms of multiplicative gaps, we explicitly present a directed graph $\bar{G}_d$ for which
\[
\chi_f (\bar{G}_d)> (2.5244) \, \chi_{\ell}(\bar{G}_d).
\]
It was recently communicated to us \cite{Simonyipersonal} that the multiplicative gap cannot exceed the constant $e$ for any directed graph $\bar{G}_d$ , i.e.
\[ 
 \chi_f (\bar{G}_d)  \leq \text{e} \, \chi_{f \ell} (\bar{G}_d).
\]

In this work, we present a proof that the ratio $\chi_f(\bar{G}_d)/\chi_{f \ell} \left( \bar{G}_d\right)$ is at most a constant, obtained in parallel to our communication \cite{Simonyipersonal}.

\subsection{Discussion of Techniques}
Linear index coding can be mapped into a \textit {vector assignment} problem. For ease of exposition, we describe the case of scalar linear index coding. The goal is to design 
$n$ vectors $v_1,v_2,\ldots v_n$ that satisfy a set of linear independencies. Specifically, 
for each vector $v_i$, one is given a set of indices $S(i)$ and a requirement that
\[ 
v_i \notin \text{span} (v_{S(i)}),
\]
where we indicate the set of vectors having indices in $S(i)$ by $v_{S(i)}$.
From the interference alignment perspective~\cite{maleki2012index}, $S(i)$ are just 
the set of indices each corresponding to a packet that the user $i$ does not have as side information. We call it the \textit{interfering set} of indices (users) $S(i)$ for user $i$.

It is trivial to satisfy these requirements if the vectors $v_i$ lie in $n$ dimensions. 
The goal of scalar linear index coding is to minimize the dimension $k$ of these vectors while maintaining 
the required linear independencies. Now, we look at a related problem, i.e. a coloring problem. A valid \textit{proper coloring} of the indices is an assignment of colors to indices such that $i$ and $S(i)$ are assigned different colors. It is possible to obtain a vector assignment from a $k$ coloring solution by simply assigning vectors from the normal basis of length $k$ to each different color. When normal basis vectors are used, the resulting assignment of vectors is called  a \textit{coloring assignment}.

Assume, for example, that $v_1$ must be linearly independent from $v_2$ and $v_3$. We would like to make $\text{span} \{v_2,v_3\}$ to have a low dimension, if possible, to make it easier for $v_1$ to be outside this subspace. A coloring assignment would achieve this by re-using the same vector for $v_2$ and $v_3$ (same color) if other linear independency constraints allow that. 

As in the above example, in this paper, we use coloring of indices as an \textit{alignment guide} to later assign vectors to colors. Clearly, in any coloring intended for an assignment later (not necessarily a coloring assignment), $i$ and $S(i)$ have different colors since they must be assigned different vectors. However, it is possible to reduce the dimension in which the vectors lie by assigning a different set of vectors, in a lower dimension, to each color if we do not restrict to just using normal basis vectors. An equivalent way of thinking about coloring assignment is through MDS codes. An $(p,q)$ Maximum-Distance Separable (MDS) code is a set of $p$ vectors of length $q$ that are in general position, \textit{i.e.} any 
$q$ of the $p$ are linearly independent. The idea is to first color the indices and then create an MDS code and assign one MDS vector for each color.  A coloring assignment would mean that the code dimension $q$ and the number of vectors $p$  are both equal to the minimum colors used for coloring the indices (the chromatic number). This is just equivalent to a diagonal identity MDS code using normal basis vectors.

The key idea to go beyond the chromatic number is to realize that the bottleneck to the dimension $p$ is not really the total number of 
colors used but the maximum number of colors in an interfering set over all interfering sets $S(i)$. Therefore, if we can assign colors to the vertices so that the number of colors 
in the most colorful interfering set, i.e., the \textit{local chromatic number},  is bounded by $k^*$, we can construct an index code of that length: first create an $(p,k^*+1)$ MDS code over a sufficiently large field where $p$ is the total number of colors used and assign a vector to each color. Although each of the $n$ indices get a vector assignment, they lie in $k^*+1$ dimensions. We, therefore, still use the proper coloring of indices as an alignment guide but \textit{the 
local chromatic number is the metric that limits the code length} or equivalently the dimension of the vectors in the assignment.  In this description, we have used coloring of indices as a proxy for coloring of vertices in a suitable graph wherein the neighborhood of a vertex, corresponding to an index, reflects the linear independency constraints on a particular index.

The remainder of this paper is structured as follows. We start with precise definitions of the index coding and graph coloring concepts we need. We then state our results that show how index codes
can be created with the appropriate lengths related to various coloring numbers. We subsequently discuss the multiplicative and additive gaps between different coloring based bounds and conclude. 

	\section{Definitions}\label{Defsec}

In an index problem, there are $n$ users each requesting a \textit{distinct} packet $x_i,~\forall i \in \{1,2,3 \ldots n\}$ from a common broadcasting agent who needs to deliver these packets with a minimum number of bits over a public broadcast channel. In addition, each user has some side information packets denoted by $S(i)$, which is a subset of packets that other users want. Let $[n]$ denote the set $\{1,2 \ldots n \}$. Here, $S(i) \subseteq [n]-{x_i}$. This index coding problem with distinct user requests and individual user side information can be represented as a \textit{directed side information graph} $G_d(V,E_d)$ where $(i,j) \in E_d$, i.e. there is a directed edge from node $i$ to node $j$ if user $i$ has packet $x_j$ as side information. Each node corresponds to a user in this digraph. Note that a more general version of index coding allows multiple users 
to request the same packet. This corresponds to hypergraphs and we do not consider it in this paper. 
In what follows, we define the minimum broadcast rate for the index coding problem. To be consistent, we follow the definitions of Blasiak \textit{et al.} \cite{blasiak2010index} very closely.
 
         \begin{definition}
			(\textit{Valid index code}) Let $x_i \in \Sigma$, where $\lvert \Sigma \rvert =2^t$ for some integral $t$. Here, $x_i$  is the packet desired by user $i$. A \textit{ valid index code} over the alphabet $\Sigma$ is a set $(\phi,\{\gamma_i\})$ consisting of: 
  \begin{enumerate}
        \item An encoding function $\phi:\Sigma^n \rightarrow \{ 0,1\}^p$ which maps the $n$ messages to a transmitted message of length $p$ bits for some integral $p$. 
        \item   $n$ decoding functions $\gamma_i$ such that for every user $i$, $\gamma_i(\phi \left( x_1,x_2 \ldots x_n \right),\{x_j\}_{j \in S(i)}) = x_i$. In other words, every user would be able to decode its desired message from the transmitted message and the side information available at user $i$.
  \end{enumerate}
            $\hfill \lozenge$ 
\end{definition}

 The \textit{broadcast rate} $\beta_{\Sigma}(G_d,\phi, \{ \gamma_i \})$ of the $(\phi,\{ \gamma_i \})$ index code is the number of transmitted bits per received message bit at every user, i.e. $\beta_{\Sigma}(G_d,\phi,\{ \gamma_i\})= \frac{p}{\log_2 \lvert \Sigma \rvert}= \frac{p}{t}$.
        \begin{definition}  
             (\textit{Minimum broadcast rate:})  The minimum broadcast rate $\beta(G_d)$ of the given problem $G_d(V,E_d)$ is the minimum possible broadcast rate of all valid index codes over all alphabets $\Sigma$, i.e. $\beta(G_d)= \inf \limits_{\Sigma} \inf \limits_{\phi,\{\gamma_i \}} \beta_{\Sigma} (G_d,\phi,\{ \gamma_i \})$. 
         $\hfill \lozenge$ 
 \end{definition}

              If $t=1$, and the encoding function $\phi:\{0,1\}^n \rightarrow \{0,1\}^p$ is a concatenation of separable linear encoding functions, i.e. $\phi= \left(\phi_1,\phi_2 \ldots \phi_p \right)$ such that $\phi_j:\{0,1\}^n \rightarrow \{0,1\}$ is linear over the input bits $x_i$, then the index code is called a valid \textit{binary linear scalar} index code. The minimum broadcast rate over all binary linear scalar index codes for a given problem represented by digraph $G_d$ is denoted $\beta_2(G_d)$. $\beta_2(G_d)$ \cite{bar2006index}  was shown to be equal to the graph parameter $\mathrm{minrank}_2(G_d)$. It is clear from the definitions that $\beta(G_d) \leq \beta_2(G_d)$. 
            
 		 $\beta(G_d)$ is the minimum broadcast rate over all index codes (linear and non-linear) of the index coding problem given by the digraph $G_d$. There has been a series of information theoretic LPs, which was developed in \cite{blasiak2010index}, whose optima is used to bound the quantity $\beta(G_d)$ from above and below. The paper by Blasiak \textit{et al.} \cite{blasiak2010index} considers index coding problems where user requests could overlap, i.e. same packet can be requested by multiple users. The above definitions of minimum broadcast rate and decodability extend to the general case, except that the problem can be represented over a directed side information hypergraph, or equivalently, by a directed bi-partite side information graph \cite{tehrani2012bipartite}, with one partition for packets and another partition for receivers. In this work, we restrict our attention to the case of digraphs where users have distinct requests.

		We recall three bounds from \cite{blasiak2010index} that are relevant to this work for the general problem with possible overlaps between user requests. The first bound is $b_2(G)$. This is the information theoretic lower bound obtained by using only sub-modularity properties of the entropy function (in information theory parlance bounds using 'Fano's Inequality'). The second one is the \textit{fractional weak hyperclique cover}, denoted by $\psi_f$. The third is the fractional strong hyperclique cover, denoted by $\bar{\chi}_f$. We note that last two bounds are identical for the case when user requests are distinct or equivalently when the problem can be represented as a directed side information graph $G_d$. Hence, $\bar{\chi}_f(G_d)= \psi_f(G_d)$. We refer the reader to \cite{blasiak2010index} for the exact definitions of hyperclique covers for the general problem. 

Here, we define $\bar{\chi}_f(G_d)$ only for a directed graph $G_d$.  For this, we need the following definitions.

\begin{figure}
  \centering
 \includegraphics[width=8.5cm]{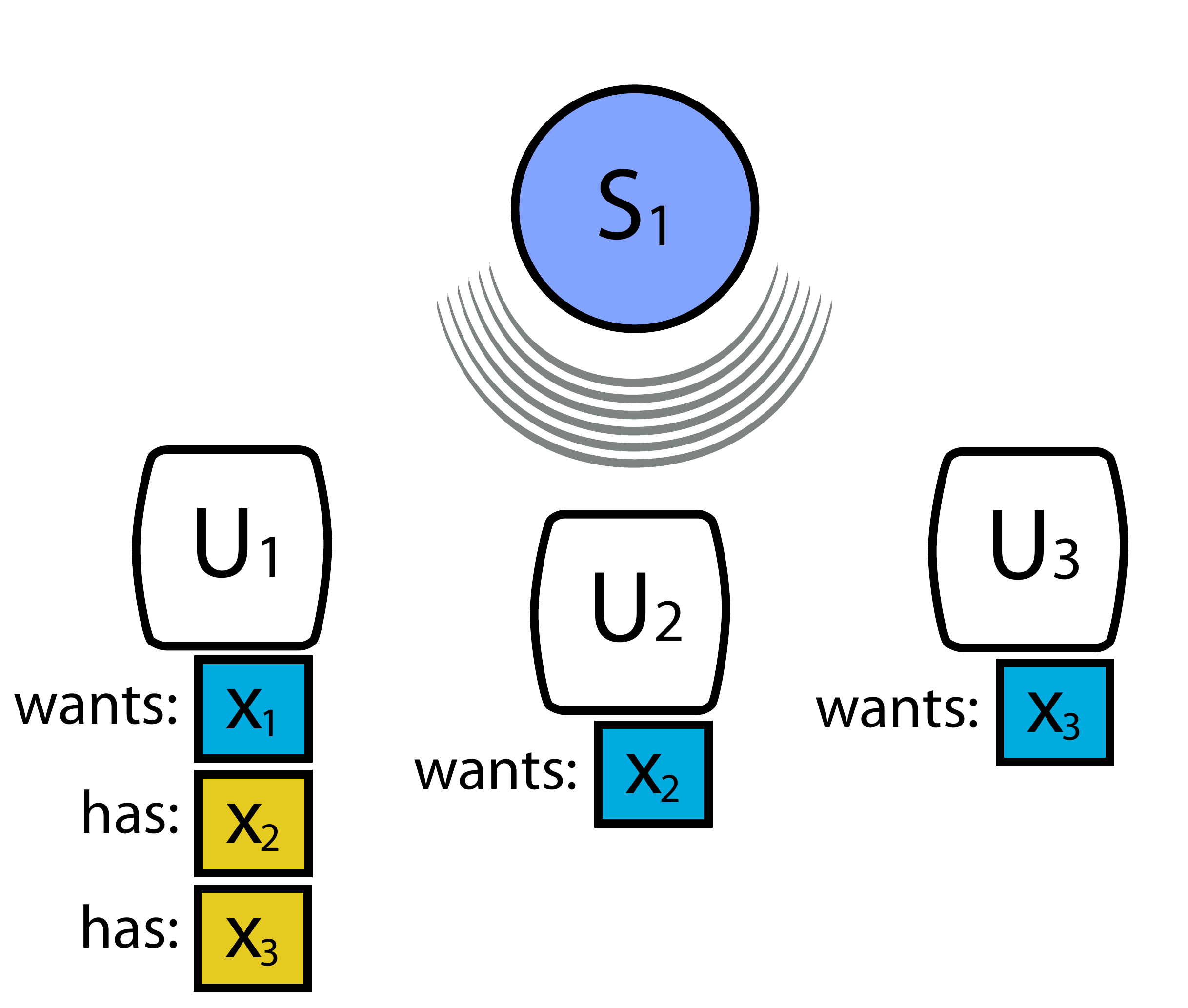}
  \caption{Index coding example: We have three users $U_1,U_2,U_3$ and a broadcasting base station $S_1$. Each user has some side information packets and requests a distinct packet from $S_1$.
The base station $S_1$ knows everything and can simultaneously broadcast to all three users noiselessly. 
User $U_i$ requests packet $x_i$. User $U_1$ has packets $x_2$ and $x_3$ as side information while users 
$U_2$ and $U_3$ have no side information. In this example three transmissions are required, so $\beta=3$.
}
\label{Graphdrawind1}
\end{figure}

\begin{figure}
  \centering
\includegraphics[width=11cm]{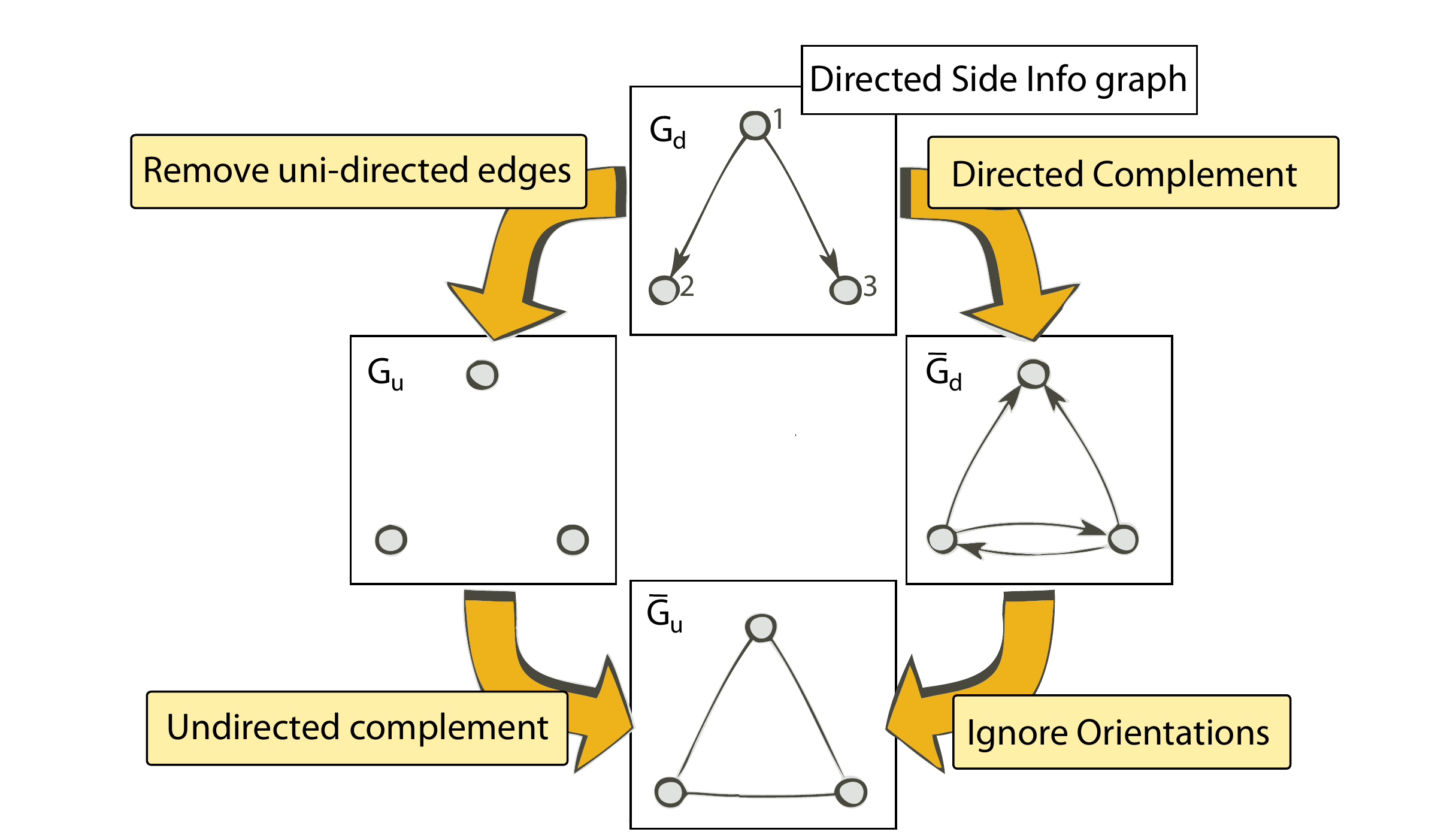}
  \caption{Index coding representation using the directed side information graph $G_d$. There are two alternate ways to reach $\bar{G}_u$. One through $G_u$, the underlying undirected side information graph. The other way is through $\bar{G}_d$, the interference graph. Clearly $\chi(\bar{G}_u)=3$.}
\label{Graphdrawind}
\end{figure}

      \begin{definition}
              (\textit{Interference graph}) The interference graph, denoted by $\bar{G_d}(V,\bar{E}_d)$ of an index coding problem is a \textit{ directed complement} of the directed side information graph $G_d$. For every vertex $i$, $(i,j) \in \bar{E}_d$ iff $(i,j) \notin E_d$. In other words, there is a directed edge from $i$ to $j$ in the directed complement (interference graph) if and only if it is not there in the directed side information graph.      
      $\hfill \lozenge$ 
\end{definition}  

      \begin{definition}
           (\textit{Underlying undirected side information graph}) Consider a directed side information graph $G_d(V,E_d)$. The \textit{underlying undirected side information graph}, denoted by $G_u(V,E_u)$, is the graph obtained by deleting uni-directed edges (i.e. $(i,j) \in E_d$ but $(j,i) \notin E_d$) and all remaining bi-directed edges are replaced by an undirected edge, denoted by $\{i,j\}$.
      $\hfill \lozenge$ \end{definition}

      We observe that the complement of the underlying undirected side information graph $G_u$ is the graph denoted by $\bar{G}_u(V,\bar{E}_u)$ which can alternatively be obtained by ignoring the orientation of the edges in the interference graph $\bar{G}_d$ (here, bi-directed edges in $\bar{G}_d$ can be replaced by a single undirected edge in $\bar{G}_u$). The various graphs associated with the index coding problem and relationships between them are illustrated in Fig. \ref{Graphdrawind}. The graphs in Fig. \ref{Graphdrawind} correspond to the index coding problem defined in Fig. \ref{Graphdrawind1}

     \begin{definition}
                 The fractional chromatic number of an undirected graph $G(V,E)$, denoted by $\chi_f(G)$, is given by the LP:
                       \begin{align}
                                     & \min \sum \limits_{I \in {\cal I}} x_I \nonumber\\
                                      \mathrm{s.t.} & \sum \limits_{I: v \in I} x_I \geq 1, ~\forall v \in V \nonumber\\
                                        \hfill & x_I \in \mathbb{R}^{+}, ~\forall I \in {\cal I}
                       \end {align}
    where ${\cal I}$ is the set of all independent sets in $G$. $\mathbb{R}^{+}$ is the set of non negative real numbers.
     $\hfill \lozenge$
 \end{definition}
    
     \begin{definition}
             $\bar{\chi}_f(G_d)$ is defined to be the fractional chromatic number, $\chi_f(\bar{G}_u)$, of the complement of the underlying undirected side information graph, denoted by $\bar{G}_u$.  
      $\hfill \lozenge$ \end{definition} 
     
      We note that the definition of $\bar{\chi}_f$ given here for directed side information graphs is equivalent to the definition of strong and weak hyper clique covers of \cite{blasiak2010index} restricted to directed graphs.

	It was shown \cite{blasiak2010index} that $\beta(G_d) \leq \bar{\chi}_f(G_d)$ and $\bar{\chi}_f(G_d)$ corresponds to an achievable binary vector coding solution to the problem. In general, $\mathrm{minrank}_2(G_d)$ and $\bar{\chi}_f(G_d)$ are incomparable for digraphs. There are examples where $\mathrm{minrank}_2(G_d) > \bar{\chi}_f(G_d)$ and vice versa.  In this sense, the fractional chromatic number is one of the best known bounds for the index coding problem.

      We consider two other graph parameters previously studied in \cite{korner2005local}\cite{simonyi2006local}, namely, the local chromatic number $\chi_\ell (\bar{G}_d)$ and the fractional local chromatic number $\chi_{f \ell}(\bar{G}_d)$ of the interference graph $\bar{G}_d$, which we show have corresponding achievable index coding schemes and hence are related to the index coding problem. We will see later that $\chi_{f \ell}(\bar{G}_d) \leq \bar{\chi}_f(G_d)$. We show that there are examples of digraphs where the inequalities are strict. Also, these local chromatic variants exploit the directionality of the directed side information graphs for the index coding problem in a way the other achievable schemes based on chromatic/clique cover numbers proposed before do not.
    
        \begin{definition}
                 (\textit{Local chromatic number:}) Denote the closed out-neighborhood (including the vertex) of a given vertex $i$ in a directed graph by $N^{+}(i)$, i.e. $j \in N^{+}(i)$ iff $(i,j)$ is a directed edge or $j=i$. The local chromatic number of the interference graph $\bar{G}_d$ (or equivalently local clique cover number of $G_d$), denoted by $\chi_{\ell}(\bar{G}_d)$, is the optimum for the following integer program:
                       \begin{align} \label{Eqn:localchromatic}
                                     & \min t \nonumber \\
                                      \mathrm{s.t.} & \sum \limits_{I: v \in I} x_I \geq 1, ~\forall v \in V  \nonumber \\
                                            \hfill &    \sum \limits_{I: N^{+}(v) \cap I \neq \emptyset} x_I \leq t  \nonumber \\
                                        \hfill & x_I \in \{ 0,1\}, ~\forall I \in {\cal I}
                       \end {align}
    where ${\cal I}$ is the set of all independent sets in the graph obtained by ignoring orientation of edges in $\bar{G}_d$, or equivalently, the independent sets in $\bar{G}_u$ (complement of the underlying undirected side information graph of $G_d$). 
         $\hfill \lozenge$ 
\end{definition}

		One can rephrase the local chromatic number in the terminology of proper coloring of a directed graph $G$. Let $c: V \rightarrow [k]$ be any proper coloring for the graph ignoring the orientation of the edges for some integer $k$. Let $\lvert c(N^{+}(i)) \rvert$ denote the number of colors in the closed out neighborhood of the directed graph taking the orientation into account. Then, 
		\[
		\chi_{\ell}(G)= \min \limits_{c} \max  \limits_{i \in V} \lvert c(N^{+}(i)) \rvert.
		\] 
In words, the local chromatic number of a directed graph $G$ is the maximum number of colors in any out-neighborhood minimized over all proper colorings of the undirected graph obtained from $G$ by ignoring the orientation of edges in $G$. An example is shown in Fig. \ref{local_example}.

	\begin{figure}
	  \centering
	 \includegraphics[width=8cm]{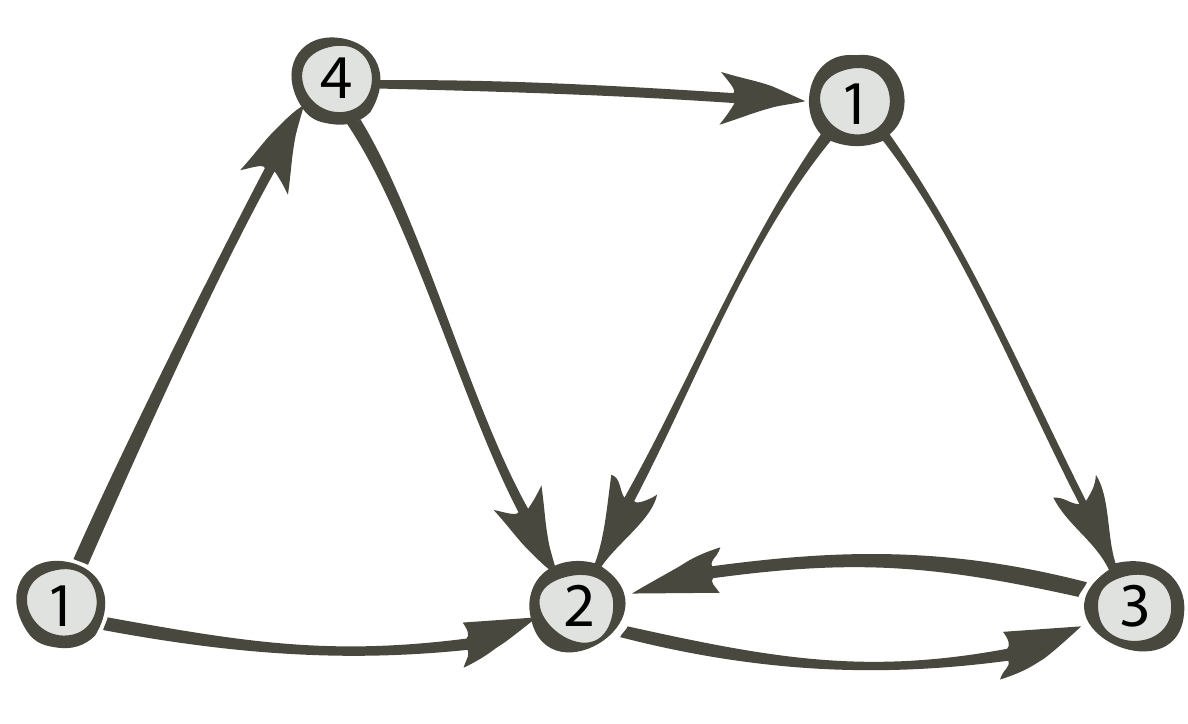}
	\caption{Example of the local chromatic number. The vertices of this directed graph have been 
	assigned the colors $\{1,2,3,4\}$. This is a proper coloring of the underlying undirected 
	graph. This coloring assignment corresponds to a local chromatic value of $3$, since the most colorful closed out-neighborhood has $3$ colors. For example, consider the unique vertex colored with color $4$: its out-neighborhood has the colors $\{1,2\}$ plus the color $4$ for the vertex itself. 
	The closed out-neighborhood of the vertex colored with color $3$ has only two colors. 
	The local chromatic number is computed by taking the minimum over proper colorings of the maximum number of colors in a closed out-neighborhood of a vertex. Note that in this graph there is 
	a proper coloring with only $3$ colors total. It is non-trivial to create graphs 
	where the local chromatic number is strictly smaller than the chromatic number.   }	
	\label{local_example}
	\end{figure}
	
      The optimum of the fractional version of the program (\ref{Eqn:localchromatic}), where $x_I$ are relaxed to be non negative reals numbers, is called the fractional local chromatic number of the interference graph, denoted by $\chi_{f \ell}(\bar{G}_d)$. In that case, the fractional local chromatic number is a result of an LP. The local chromatic number was initially defined in \cite{erdos1986coloring}. More recently, both fractional local chromatic number and the local chromatic number have been studied in \cite{korner2005local} in the context of bounds for Shannon and Sperner graph capacities.

  \section{Achievability for local chromatic numbers}
			  In this section, we show that $\chi_{\ell}(\bar{G}_d)$ and $\chi_{f \ell}(\bar{G}_d)$ correspond to linear and vector linear achievable schemes over higher fields. We have the following achievability results.

     \begin{theorem}\label{Lem:Localachiev}
                Given a directed side information graph $G_d$,
\[ 
  \beta(G_d) \leq \chi_{\ell}(\bar{G}_d).  
\]
    \end{theorem}
      \begin{proof}
                As in the integer program $(\ref{Eqn:localchromatic})$, let ${\cal I}$ denote the family of independent sets in the undirected graph $\bar{G}_u$ or the graph obtained from $\bar{G}_d$ by ignoring the orientation of the edges. Coloring this graph, involves assigning $0$'s  and $1$'s to the independent sets in the graph. Let ${\cal J} \subseteq {\cal I}$ be the set of color classes in the optimal local coloring. Let $\chi_\ell(\bar{G}_d)$ be the local coloring number. Let $J:V \rightarrow {\cal J}$ be the coloring function. To each color class (independent set assigned $1$), we assign a column vector from $\mathbb{F}_q^m$ of a suitable length $m$ and over a suitable field $\mathbb{F}_q$ by a map $\mathbf{b}:{\cal J} \rightarrow \mathbb{F}_q^{m}$.  If the message desired by each user is from the finite field $\mathbb{F}_q$, i.e. $x_i \in \mathbb{F}_q,~\forall i \in V$, then we transmit the vector 
\[
\left[\mathbf{b}(J(1)),\mathbf{b}(J(2)) \ldots \mathbf{b}(J(n))\right] [x_1,x_2 \ldots x_n]^T.
\]
Clearly the length of the code is $m$ field symbols. If the index code is valid, then the broadcast rate is $m$. 

               We now exhibit a mapping $\mathbf{b}$ with $m=\chi_{\ell}(\bar{G}_d)$ and $q \geq \lvert \cal J \rvert$. Let the colors classes in ${\cal J}$ be ordered in some way. Consider the generator matrix $\mathbf{G}$ of a $(\lvert {\cal J} \rvert, \chi_{\ell}(\bar{G}_d) )$ MDS code over a suitable field $\mathbb{F}_q$ ,where $q \geq \lvert \cal J \rvert $. For instance, Reed Solomon code constructions could be used. Assign the different columns of $\mathbf{G}$ to each color class, i.e. $\mathbf{b}(j) = \mathbf{G}_j,~ \forall j \in {\cal J} $ where $\mathbf{G}_j$ is the j-th column. Under this mapping $\mathbf{b}$ and the previous description of the index code, it remains to be shown that this is a valid code. For any vertex $i$, the closed out-neighborhood $N^{+}(i)$ contains $\lvert J \left( N^{+} (i) \right) \rvert$ colors. Because, the coloring $J$ corresponds to the optimal local coloring, there are at most $m$ colors in any closed out neighborhood. Therefore,  $\lvert J \left(N^{+}(i) \right) \rvert \leq \chi_\ell(\bar{G}_d) = m$. 
	
		    Every vertex (user) $i$ must be able to decode its own packet $x_i$. User $i$ possesses packets $x_k$ as side information when $k$ is not in the closed out-neighborhood $N^{+}(i)$ in the interference graph $\bar{G}_d$. Hence, $\mathbf{b}(J(k)) x_k$ can be canceled for all $k \notin N^{+}(i)$. The only interfering messages for user $i$ are $\{\mathbf{b}(J(k)) x_k \}_{k \in N^{+}(i)-\{i \}}$. If we show that $b(J(i))$ is linearly independent from all $\{\mathbf{b}(J(k))\}_{k \in N^{+}(i)-\{i \}}$, then user $i$ would be able to decode the message $x_i$ from its \textit{interferers} in $N^{+}(i)- \{i\}$.

		 Since ${\cal J}$ represents a proper coloring over $\bar{G}_u$,  $\mathbf{b}(J(i))$ is different from $\mathbf{b}(J(k))$ for any $k \in N^{+}(i)- \{ i\}$. Also, any $m$ distinct vectors are linearly independent by the MDS property of the generator $\mathbf{G}$. Since, $\lvert J \left(N^{+}(i) \right) \rvert \leq \chi_\ell(\bar{G}_d)=m$, i.e. the number of colors in any closed out-neighborhood is at most $\chi_{\ell}$, the distinct vectors in any closed-out neighborhood are linearly independent . This implies that $b(J(i))$ is linearly independent from $\{\mathbf{b}(J(k))\}_{k \in N^{+}(i)-\{i \}}$.  Hence, every user $i$ would be able to decode the message it desires. Hence it is a valid index code and the broadcast rate is $\chi_{\ell}(\bar{G}_d)$. 	
      \end{proof}

In the previous construction we needed a field size at least equal to the total number of colors used 
in the local coloring. A natural question is if a similar idea can be applied when we restrict the problem to binary index codes. We show that adding a logarithmic additive overhead in the code length suffices to make the binary vectors to be in general position. This leads to the following corollary:

   \begin{corollary}
               Given a directed side information graph $G_d$, $\beta_2(G_d) \leq \chi_{\ell}(\bar{G}_d)+ 2 \log n$
   \end{corollary}
     \begin{proof}
            Consider the same coloring ${\cal J}$ which yields the minimum local coloring number for $\bar{G}_d$ as in proof of Theorem \ref{Lem:Localachiev}. From the proof of Theorem \ref{Lem:Localachiev}, it is enough to design a binary matrix, $\left[\mathbf{b}(J(1)),\mathbf{b}(J(2)) \right. $ $ \left. \ldots \mathbf{b}(J(n))\right]$ of dimensions $\left( \chi_{\ell} \left(\bar{G}_d \right)+ 2\log n \right)  \times n $ such that $b(J(i))$ is linearly independent from all $\{\mathbf{b}(J(k)):  $ $~ k \in N^{+}(i)-\{i \}\}$ over the binary field for every $i$. This would be a valid binary scalar linear code and hence from the definitions, the statement of the theorem follows.

	We show a randomized construction. First, we construct a random binary matrix $\mathbf{G}$ of dimensions $\left( \chi_{\ell}+2 \log n \right) \times \lvert {\cal J}\rvert$ with i.i.d Bernoulli variables with probability of $0$ being 1/2. Then, we assign the $j-$ th column of $\mathbf{G}$ to $\mathbf{b}(j), ~\forall j \in {\cal J} $ as before. Let $ \chi_{\ell} (\bar{G}_d)= p$. We analyze the probability that some specific $p$ columns have full rank. It is shown in \cite{shokrollahi2006raptor} that the probability that a $\left( \chi_{\ell}+2 \log n \right) \times \chi_\ell$ random i.i.d matrix ( probability of each entry in the random matrix being equal to $0$ is $1/2$), is not full rank is at most $2^{-2 \log n} =  1/n^2$. 

\begin{align}
    \mathrm{Pr} \left( \exists i: ~ \mathbf{b}\left( J(i) \right) \mathrm{linearly~ dependent~ on} \{\mathbf{b}(J(k))\}_{k \in N^{+}(i)-\{i \}}   \right) & \leq n \mathrm{~Pr} \left( \mathbf{b}(1) \mathrm{~ linearly~ dependent~ on}  \right. \nonumber \\
         &  \left.  \left[ \mathbf{b}(2) \ldots \mathbf{b}(p)\right]   \right) \\
     \hfill  & \leq n \mathrm{~Pr} \left( \mathrm{a~random} \left( \chi_{\ell}+2 \log n \right) \times \chi_\ell  \right. \nonumber \\
               & \left. \mathrm{~matrix ~is~ not ~ full ~rank}    \right) \\ 
        \hfill             &=  n \frac{1}{n^2} = \frac{1}{n}                  
\end{align}
     The first inequality is because: 1) every vertex has at most $p-1$ colors in the neighborhood 2) the column vector assigned to color $i$, i.e. $\mathbf{b}(i)$, is generated independently and identically and 3) the union bound is applied. The last inequality follows from the bound for the probability that a random matrix is not full rank in the preceding discussion.

Hence, the probability that this random construction results in an invalid index code is $1/n$. Hence, a valid binary index code of length $\chi_{\ell}(\bar{G}_d)+ 2 \log n$ exists with very high probability.
    \end{proof}

The above corollary shows that, if field size is binary, with some extra length of $2 \log n$ over and above $\chi_{\ell}(\bar{G}_d)$ in Theorem \ref{Lem:Localachiev}, we can find a good index code.
      \begin{theorem}
                Given a directed side information graph $G_d$, 
\[
\beta(G_d) \leq \chi_{f \ell}(\bar{G}_d).
\]
\end{theorem}
      \begin{proof}
                	We combine the arguments of \cite{blasiak2010index} for the achievability result for the fractional chromatic number and the above arguments of using MDS codes for the local chromatic number to produce an achievable index coding scheme for the fractional local chromatic number.  For the interference graph, $\bar{G}_d$, let $\chi_{f \ell}(\bar{G}_d) = t$ for the relaxed version of the program (\ref{Eqn:localchromatic}). This implies :
         \begin{equation}                     
      x_I \geq 0,~x_I \in \mathbb{R}^{+},~ \forall I \in {\cal I},~ \mathrm{and}~~\sum \limits_{I \cap N^{+}(v)\ne \phi} x_I \leq t, ~ \forall v \in V ~\mathrm{and}~~ \sum \limits_{I:v \in I} x_I \geq 1, ~ \forall v \in V.
           \end{equation}
       Here, we recall that ${\cal I}$ is a collection of all independent sets in the undirected graph, denoted $\bar{G}_u$, obtained from $\bar{G}_d$ ignoring the orientations of the edges. Since the fractional version of problem (\ref{Eqn:localchromatic}) has only integer coefficients, all $x_I$ are rational. Let $r$ be the least common multiple of all the denominators of the rationals $x_I$. Then we can redefine: $y_I = r x_I$. The equations become: 
        \begin{equation}                     
      \sum \limits_{I \cap N^{+}(v)\ne \phi} y_I \leq rt=s, ~ \forall v \in V ~\mathrm{and}~~ \sum \limits_{I:v \in I} y_I \geq r, ~ \forall v \in V.
           \end{equation}
      Let $p = \sum \limits_{I} y_I$. As $y_I$ is integral, we will call $y_I$ the number of 'colors' assigned to $I$. Since a subset of an independent set is independent, one can make sure all the inequalities $\sum \limits_{I:v \in I} y_I \geq r$ are equalities by  carrying out the following operation, repeatedly for any vertex $v$ for which $\sum \limits_{I:v \in I} y_I > r $,  by choosing a suitable $I$: Choose the subset $J = I - \{v\}$ and increase $y_J$ by 1 and decrease $y_I$ by 1 if $y_I>0$. This operation would not affect any of the inequalities including the locality constraints. For the locality constraints, the colors in any closed out-neighborhood will only reduce or remain the same due to this operation as colors for a bigger independent set are assigned to a smaller independent set contained in it.  

 After these operations, if $y_I >1$, one may consider a \textit{collection} of independent sets ${\cal I}$ where independent sets are repeated but each $y_I =1$. Until now, the arguments remain similar to that in \cite{blasiak2010index} with the exception of taking note of the newly added locality constraints.

      We have a sequence of $p$ independent sets (possibly repeated) and every vertex is in $r$ of them and the out-neighborhood sees at most $s$ of them. Clearly, $s \geq r$. Like in the proof of Theorem \ref{Lem:Localachiev}, we generate a $(p,s)$ MDS code over a field of size greater than $p$. To every independent set in the possibly repeated sequence of $p$ sets, we assign one column of the generator matrix. Let ${\cal I}(v)$ denote the independent sets (possible repeated) that contain the vertex $v$. Under the assignment of columns to the independent sets (possibly repeated), all columns assigned to independent sets in ${\cal I}(v)$ are collected in the matrix $\mathbf{B}(v)$. Since every vertex is a part of exactly $r$ independent sets (possible repeated), $\mathbf{B}(v)$ is an $s \times r$ matrix. Since $r \leq s$, the columns of $\mathbf{B}(v)$ are linearly independent by the MDS property. Since the number of independent sets intersecting $N^{+}(v)$ is at most $s$, for any vertex $v$, columns of  $\mathbf{B}(v)$ are linearly independent from the columns assigned to neighborhood vertices, given in a concatenated form as $\left[ \mathbf{B}(u_1)~\mathbf{B}(u_2) \ldots \mathbf{B}(u_k) \right]$ where $\{u_1~u_2 \ldots u_k\} = N^{+}(v) -\{v\}$. 

	Let each message $\mathbf{x}_i \in \mathbb{F}_q^{r \times 1}$ be a vector of $r$ field symbols where $q \geq p$. The index code is given by: $\left[\mathbf{B}(1)~ \mathbf{B}(2) \ldots \mathbf{B}(n) \right] \left[\left(\mathbf{x}_1\right)^T ~\left(\mathbf{x}_2\right)^T \ldots \left(\mathbf{x}_n\right)^T   \right]^T$. By the MDS property of the code used and by previous arguments, this code is a valid index code. The broadcast rate is $s/r = t = \chi_{f \ell}(\bar{G}_d)$. 
       \end{proof}

\section{Multiplicative Gap between fractional and local chromatic variants}
          From the definition of the fractional local chromatic number, due to additional locality constraints, it is clear that for any directed graph $G_d$, $\chi_{f \ell}(\bar{G}_d) \leq \chi_f (\bar{G}_u) = \bar{\chi}_f(G_d)$. In other words, the fractional chromatic number uses just the information contained in the bi-directed edges of the side information graph. Through locality, the fractional local chromatic number exploits the directionality of certain uni-directed edges which are neglected by the fractional chromatic number. If the problem is over a bi-directed side information graph $G_d$ ($i$ has $x_j$ iff $j$ has packet $x_i$) or an undirected side information graph $G$, then, from results in \cite{korner2005local}, we have the following lemma:

\begin{lemma}
      \cite{korner2005local} $ \chi_{f \ell}(\bar{G}) = \chi_f(\bar{G}) \leq \chi_{\ell}(\bar{G}) \leq \chi (\bar{G})$
\end{lemma}
         This means for a problem on a bi-directed side information graph, the fractional local chromatic number is identical to the fractional chromatic number. 
%
 Therefore, we consider digraphs which have some uni-directed edges. We give an example of a digraph (which is not bi-directed) over $n$ edges where the difference between fractional and local chromatic number is $\Theta(n)$. Consider a complete undirected graph on $n$ vertices. Number the vertices from $1$ to $n$. Further, label the vertices 'odd' or 'even' depending on the number assigned to the vertices. Now, we construct a directed graph from the complete graph by orienting every edge in exactly one direction or the other. Consider two vertices numbered odd, i.e. $2p+1$ and $2q+1$. If $p <q$, there is an edge directed from vertex numbered $2p+1$ to the vertex numbered $2q+1$. Consider two vertices numbered $2p$ and $2q+1$. If $2p > 2q+1$, then there is a directed edge from vertex numbered $2p$ to the vertex numbered $2q+1$. If $2p < 2q+1$, then there is a directed edge from vertex numbered $2q+1$ to the vertex numbered $2p$.  Consider two vertices numbered even, i.e. $2p$ and $2q$. If $p <q$, there is a directed edge from vertex numbered $2p$ to the vertex numbered $2q$.  Consider any vertex numbered odd. It has edges directed outwards towards all odd vertices bigger than itself and all even vertices smaller than itself in the numbering. Similarly, consider any vertex numbered even. It has edges directed outwards towards all odd vertices smaller than itself and all even vertices bigger than itself.  Therefore, the out-degree  of every vertex is bounded by $\frac{n}{2}+1$. Hence, the local chromatic number of the directed version is at most $\frac{n}{2} +1 $ and the fractional chromatic number is $n$.

We mention that this example is a special case of a more natural generalization of a relation between the local and chromatic number for any graph and its orientations given in Proposition $4$ in \cite{simonyi2010directed}.  Hence, the achievable index coding schemes proposed here based on fractional local chromatic number is strictly better in terms of broadcast rate than the ones proposed in \cite{blasiak2010index}. We consider the question of multiplicative gap. We came to know of a parallel work in progress \cite{Simonyipersonal} that has established a tighter upper bound of $e$ for the ratio $\chi_f(\bar{G}_d)/\chi_{f\ell} (\bar{G}_d)$.  In this work , using results regarding graph homomorphisms, we prove that the ratios $\chi_f(\bar{G}_d)/\chi_\ell(\bar{G}_d)$ and $\chi_f(\bar{G}_d)/\chi_{f \ell}(\bar{G}_d)$ are upper bounded by a constant. This result is obtained in parallel to \cite{Simonyipersonal}.

\subsection{Universal upper bound for $\chi_f(\bar{G}_d)/\chi_\ell(\bar{G}_d)$}
 We upper bound the ratio between the fractional chromatic number $\chi_f(\bar{G}_u)=\chi_f(\bar{G}_d)$ and the local chromatic number $\chi_\ell(\bar{G}_d)$ by $\frac{5}{4}\text{e}^2$ for any directed graph $\bar{G}_d$.

\begin{theorem}\label{localbndtheorem}
    $\chi_f(\bar{G}_d)/ \chi_\ell (\bar{G}_d) \leq \frac{5}{4}\text{e}^2$  
$\hfill \square$ \end{theorem} 

	Before proving Theorem \ref{localbndtheorem}, we would like to recall notations and present a few technical lemmas. We always work on the interference graph $\bar{G}_d$ which is the directed complement of the directed side information graph $G_d$. $G_u$ is obtained by throwing away uni-directed edges $(i,j) \in E_d,~(j,i) \notin E_d$ and replacing bi-directed edges $(i,j),(j,i) \in E_d$ by an undirected edge $\{ i,j\}$. $\bar{G}_u$ is the complement of $G_u$. Alternatively, $\bar{G}_u$ is obtained by ignoring orientations of the directed edges in the interference graph $\bar{G}_d$. If there is a bi-directed edge in $\bar{G}_d$, only one undirected edge is used to replace them in $\bar{G}_u$.  
	
	From discussions in the previous sections, $\bar{\chi}_f(G_d)=\bar{\chi}_f(G_u)=\chi_f(\bar{G}_u)=\chi_f(\bar{G}_d)$. We intend to upper bound the ratio $\chi_f(\bar{G}_d)/ \chi_\ell (\bar{G}_d)$ for any $G_d$.  We use ideas of graph homomorphisms and universal graphs already defined in \cite{korner2005local}.
       Let us review those concepts. A directed graph $G_d$ is homomorphic to another directed graph $H_d$ if there is a function  $f: V(G_d) \rightarrow V(H_d)$ such that if $(i,j) \in E(G_d)$ , then $(f(i),f(j)) \in E(H_d)$. In other words, the directed edges are preserved under the mapping $f$.  Similarly, an undirected graph $G$ is homomorphic to another undirected graph $H$ if there is a function $f: V(G) \rightarrow V(H)$ such that $\{i,j\} \in E(G)$ implies $\{f(i),f(j)\} \in E(H)$. 

     Let $G_u$ be the undirected graph obtained by ignoring orientation of the directed edges in $G_d$ (and replacing any bi-directed edge by a single undirected edge). Let $H_u$ be obtained similarly from $H_d$. Then, 
  \begin{lemma}\label{undhomo}
              If $G_d$ is homomorphic to $H_d$, then $G_u$ is homomorphic to $H_u$. 
 \end{lemma}    
 \begin{proof}
        If $\{i,j \} \in G_u$, then either $(i,j) \in E(G_d)$ or $(j,i) \in E(G_d)$ or both. This implies either $(f(i),f(j)) \in E(H_d)$ or $(f(j),f(i)) \in E(H_d)$. This means that $ \{ f(i),f(j) \} \in E(H_u)$.  Hence $G_u$ is homomorphic to $H_u$. 
\end{proof}

 Let us recall the definitions of  the universal graph $U_d(m,k)$ from \cite{korner2005local}. Let $[m]$ denote $\{1,2,3 \ldots m\}$.
\begin{definition} 
     $V\left( U_d(m,k) \right) = \{(x,A): x \in [m], A \subseteq [m], \lvert A \rvert= k-1, x \notin A \} $.
    $E \left(U_d (m,k) \right) = \{ \left( (x,A) , (y,B) \right): y \in A \}$. In other words, there is a directed edge from $(x,A)$ to $(y,B)$ if $y \in A$.
$\hfill \lozenge$
 \end{definition}
 
With some different notation from \cite{korner2005local} , let us define the undirected graph $U(m,k)$ obtained by disregarding the orientation of the edges in $U_d(m,k)$. It is defined as:

  \begin{definition} 
     $V\left( U (m,k) \right) = \{(x,A): x \in [m], A \subseteq [m], \lvert A \rvert= k-1, x \notin A \} $.
    $E \left(U (m,k) \right) = \{ \left( (x,A) , (y,B) \right): y \in A ~\mathrm{or}~ x \in B \}$. In other words, there is an undirected edge between $(x,A)$ and $(y,B)$ if $y \in A$ or $x \in B$ or both.
$\hfill \lozenge$ 
\end{definition}

It has been shown in \cite{korner2005local} that if $\chi_\ell ( \bar{G}_d ) = k$, then there is an integer $m \geq k$ such that $\bar{G}_d$ is homomorphic to $U_d(m,k)$ and $\chi_\ell \left( U_d(m,k) \right) = k$. Since $\bar{G}_d$ is homomorphic to $U_d(m,k)$, from Lemma \ref{undhomo}, $\bar{G}_u$ is homomorphic to $U(m,k)$. It is known from \cite{godsil2001algebraic}, that if undirected $G$ is homomorphic to undirected $H$, then $\chi_f(G) \leq \chi_f(H)$. Then clearly, $\chi_f(\bar{G}_d)/\chi_\ell(\bar{G_d}) \leq \chi_f (U(m,k))/k$ for all $\bar{G}_d:\chi_\ell(\bar{G}_d) =k$. We will show that $\chi_f(U(m,k))/k$ is upper bounded by $\text{e}^2$ for all parameters. 

\begin{proof}[Proof of Theorem \ref{localbndtheorem}]
We show that $\chi_f \left( U(m,k) \right)/k \leq \frac{5}{4}\text{e}^2, k \geq 2$.  It is easy to observe that $U(m,k)$ is a vertex transitive graph. For vertex transitive graphs, it is known that $\chi_f \left ( U(m,k) \right) = \frac{ \lvert V \left( U(m,k) \right)\rvert } { \alpha \left( U (m,k)\right) }$, where $\alpha (.)$ is the independence number of the graph. We analyze the structure of maximal independent sets of the graph $U(m,k)$. Let $I=\{ (x_1,A_1), (x_2,A_2) \ldots (x_n,A_n) \}$ be an independent set. Let  $X = \{x_1,x_2 \ldots x_n \} $ and $A = \bigcup \limits_{i} A_i$. Then clearly, from the definition of $U(m,k)$, $A \cap X = \emptyset$, otherwise there is an edge between two vertices in the set $I$. If $X \cup A \neq [m]$, then one can add elements $(x_i,A_i)$ to $I$ s.t. $x_i \in X$ and $A_i \in A' = [m]- X$. Let $\lvert A' \rvert= m-p$ and $\lvert X \rvert= p$. Choose the set $I'$ containing all possible vertices $(x,B)$ such that $x \in X $ and $B \subseteq A',~ \lvert B \rvert = k-1$. $I'$ is independent and $I \subseteq I'$. $\lvert I'\rvert= p {m-p \choose k-1 }$. 

    Hence, every independent set is contained in a suitable independent set $I'$ parametrized by $p$. Therefore, $\alpha\left( U(m,k) \right) = \max \limits_{1 \leq p \leq m-k+1} p {m-p \choose k-1}$. $\lvert V \left( U(m,k) \right) \rvert = m {m-1 \choose k-1}$.

Let $m'=\lfloor m/k\rfloor * k$. Clearly, $m \leq m'+k$. If $m \geq 4k$, then $m' \geq 4k$. For $m \geq 4k$, we have:
 
  \begin{align}\label{eq1}   
       \chi_f \left( U(m,k) \right)/k  & =  \min \limits_{1 \leq p \leq m-k+1} \frac{m {m-1 \choose k-1}}{ k p {m-p \choose k-1}} \\ \label{eq1.1}
                                              & \leq    \min \limits_{1 \leq p \leq m-k+1} \frac{\left(m'+k \right) {m-1 \choose k-1}}{ k p {m-p \choose k-1}}  \\ \label{eq1.2}
                                              &  \leq    \min \limits_{1 \leq p \leq m-k+1} \frac{\left(m'+k \right) {m'-1 \choose k-1}}{ k p {m'-p \choose k-1}}  \\  \label{eq1.3} 
                                              &  \leq   \frac{\left(m'+k \right) {m'-1 \choose k-1}}{ k p {m'-p \choose k-1}} \rvert_{p = m'/k} \\
                                            & \leq   \left(1+\frac{k}{m'}\right)  \frac{ (m'-1) (m'-2) \ldots (m'-k+1)} { \left( m' \left(1-1/k\right) \right) \left( m' \left(1-1/k\right) -1 \right) \left( m' \left(1-1/k\right) -2 \right) \ldots  \left( m' \left(1-1/k\right) - k+2 \right)  } \\ \label{eq1.4}
                                             & \leq  \frac{5}{4} \frac{m'-k+1}{\left( m' \left(1-1/k\right) \right)}  \frac{ (m'-1) (m'-2) \ldots (m'-(k-2))} { \left( m' \left(1-1/k\right) -1 \right) \left( m' \left(1-1/k\right) -2 \right) \ldots  \left( m' \left(1-1/k\right) - (k-2) \right)  } \\ \label{eq2}
                                              & \leq \frac{5}{4}   \frac{1}{\left(1-1/k\right)}  \frac{ (m'-1) (m'-2) \ldots (m'-(k-2))} { \left( m' \left(1-1/k\right) -1 \right) \left( m' \left(1-1/k\right) -2 \right) \ldots  \left( m' \left(1-1/k\right) - (k-2) \right)  } \\ \label{eq3}
                                           & \leq   \frac{5}{4} \frac {1}{\left(1-1/k\right)}  \frac{1}{\left(1-1/k \right)^{2(k-2)}} \\   
                                          & \leq   \frac{5}{4} \frac {1}{\left(1-1/k \right)^{2(k-1)}} \\ 
        	                                     & \leq \frac{5}{4} \text{e}^2   \label{eq4}
  \end{align} 
     (\ref{eq1.1}) is because $m \leq m'+k$. (\ref{eq1.2}) is because $\frac{{m-1 \choose k-1}}{ {m-p \choose k-1}}$ is a decreasing function with respect to the variable $m$. To see that, $\frac{{m-1 \choose k-1}}{ {m-p \choose k-1}} = \prod \limits_{i=1}^{k-1} \frac{m-i}{m-(p-1)-i}$. But, $\frac{m-i}{m-(p-1)-i} = 1+\frac{p-1}{m-(p-1)-i}$ which is decreasing in $m$. (\ref{eq1.3}) is because substituting any value for $p$ gives an upper bound. Here, we choose $p =m'/k$. (\ref{eq2}) is because $m-k+1 \leq m,~ \forall k \geq 1$. (\ref{eq3}) is valid under the assumption that $m/k \geq 4, k \geq 2$ and Lemma \ref{esqlemma} given below. When $m \geq 4k$, $m' \geq 4k$. Therefore, (\ref{eq1.4}) is valid.  (\ref{eq4}) is because $\left( 1+\frac{1}{k-1} \right)^{k-1} \leq \text{e},~\forall k \geq 1$. We now prove Lemma \ref{esqlemma} needed to justify (\ref{eq3}) below.

 If $m/k \leq 4$ then taking $p=1$ in (\ref{eq1}), we can show that the ratio is upper bounded by $m/k \leq 4$.   
\end{proof}
We note that $k \geq 2$ is not a restriction, since, $\chi_\ell \geq 2$ for a digraph with at least one directed edge .

\begin{lemma} \label{esqlemma}
$  \frac{m-i}{m \left(1-1/k \right)- i } \leq \frac{1}{\left (1-1/k \right)^2}, ~ \forall 1 \leq i \leq k, ~ m \geq 4k,~ k \geq 2$.
\end{lemma}
 \begin{proof}
     The following equivalence hold for $x \geq y$: 
            \begin{equation}
                \frac{x-i}{y-i} \leq \frac{x^2}{y^2} \Leftrightarrow i (x+y) \leq xy
            \end{equation}
 It is sufficient to show that $i \leq xy/(x+y)$ when $i \leq k$, $x=m,~y=m \left(1-1/k \right)$ . This means it is enough to show $ k \leq m \frac{\left(1-1/k \right)}{2}$. Since $k \geq 2$, $(1-1/k) \geq 1/2$, for $ m \geq 4k$,  $ k \leq m \frac{\left(1-1/k \right)}{2}$ holds. Hence, the statement in the lemma is proved.  
\end{proof}
 We provide a numerical example where the ratio is larger than $2.5$. Consider $U_d (281,9)$. Computer calculations show that $\chi_f\left(U_d ( 289,9)\right)/ \chi_{\ell}\left( U_d ( 289 ,9 )\right) = 2.5244$. 

In the next subsection, we extend this framework to bounding the ratio between fractional chromatic and the fractional local chromatic number.

\subsection{Universal upper bound for $\chi_f(\bar{G}_d)/\chi_{f \ell}(\bar{G}_d)$}
       In this subsection, we prove the following result about multiplicative gap between the fractional chromatic and fractional local chromatic number.
\begin{theorem}\label{fraclocalbndthm}
   $\chi_f(\bar{G}_d)/\chi_{f\ell} (\bar{G}_d)  \leq \frac{5}{4}\text{e}^2$
$\hfill \square$ \end{theorem}
	Since we work on the interference graph, $\bar{G}_d$ will be used throughout. The proof of the theorem follows the same recipe as the proof of the bound for the ratio between fractional chromatic and the local chromatic number, i.e. the ratio $\chi_{f}\left(\bar{G}_d \right)/ \chi_{l} \left(\bar{G}_d\right)$. Apart from the $LP$ characterization of the fractional local chromatic number, there is another characterization using $r- $fold local colorings of a directed graph $\bar{G}_d$ \cite{korner2005local}. We will review the results regarding $r-$ fold colorings from \cite{korner2005local}.

\begin{definition}
        (\textit{r-fold local coloring number of a digraph}) A proper $r$-fold coloring of $\bar{G}_d$ is a coloring of $\bar{G}_u$ (the undirected graph obtained by removing edge orientations in $\bar{G}_d$) where $r$ distinct colors are assigned to each vertex such that color set of two adjacent vertices are disjoint. The $r$-fold local coloring number of a graph, denoted by $\chi_\ell(\bar{G}_d,r)$, is less than or equal to $k$ if there is a proper $r$-fold coloring of $\bar{G}_d$ such that the total number of distinct colors in the closed out-neighborhood of any vertex is at most $k$. 
$\hfill \lozenge$ 
\end{definition}
   
$\chi_\ell \left( \bar{G}_d,r \right)$ is the minimum possible maximum number of colors in any closed out neighborhood over all possible proper $r$-fold colorings of graph $\bar{G}_d$. It is known that \cite{korner2005local}: 
   \begin{equation}
          \chi_{f\ell} (\bar{G}_d)= \inf \limits_{r} \chi_\ell(\bar{G}_d,r)/r   
  \end{equation}
 Let us define a directed universal graph $U_d(r,m,k)$ as follows: 
 
\begin{definition}      
     $V\left( U_d(r,m,k) \right) = \{(X,A): x \subseteq [m],~ A \subseteq [m],~ \lvert A \rvert= k-r,~ \lvert X \rvert = r,~ X \cap A = \emptyset \} $.
    $E \left(U_d (r,m,k) \right) = \{ \left( (X,A) , (Y,B) \right): Y \subseteq A \}$. In other words, there is a directed edge from $(X,A)$ to $(Y,B)$ if $Y \subseteq A$.
$\hfill \lozenge$
 \end{definition}

Consider the undirected graph $U(r,m,k)$ obtained by ignoring the orientation of the directed edges in $U_d(r,m,k)$ and bi-directed edges replaced by a single undirected edge.  It is formally defined as:

 \begin{definition}      
     $V\left( U(r,m,k) \right) = \{(X,A): x \subseteq [m],~ A \subseteq [m],~ \lvert A \rvert= k-r,~ \lvert X \rvert = r,~  X \cap A = \emptyset \} $.
    $E \left(U(r,m,k) \right) = \{ \left( (X,A) , (Y,B) \right): Y \subseteq A \ \mathrm {or}\  X \subseteq B  \}$. In other words, there is a directed edge from $(X,A)$ to $(Y,B)$ if $Y \subseteq A$ or $X \subseteq A$ or both.
$\hfill \lozenge$
 \end{definition}

Replacing all undirected edges by directed edges and out-neighborhoods by closed out-neighborhoods in the argument of Lemma 3 in \cite{korner2005local}, we have the following lemma:
 \begin{lemma} \label{homolem}
           Given $\bar{G}_d$, if $\chi_\ell \left(\bar{G}_d,r\right) \leq k$ then there is an $m$ such that $\bar{G}_d$ is homomorphic to $U_d(r,m,k)$. 
 \end{lemma}

This means, in the undirected sense $\bar{G}_u$ is homomorphic to $U(r,m,k)$, by Lemma \ref{undhomo}, if $\bar{G}_d$ is homomorphic to $U_d(r,m,k)$ in the directed sense.  Also, $\chi_f(\bar{G}_d) \leq \chi_f \left( U(r,m,k) \right)$ as fractional  chromatic numbers cannot decrease under homomorphism.

Consider $\bar{G}_d$ such that $\chi_{\ell}(\bar{G}_d,r)= k$ for some $r$ and $k$. By Lemma \ref{homolem}, $\bar{G}_d$ is homomorphic to $U_d(r,m,k)$ for some $m$. Hence, $\chi_f(\bar{G}_d)=\chi_f(\bar{G}_u) \leq \chi_f \left( U (r,m,k) \right)$. Here, $k \geq 2$ as it is a $r$-fold local chromatic number of a directed graph with one directed edge.  Now, we have the following lemma: 

\begin{lemma}
   $\chi_f( \bar{G}_d )/ \frac{k}{r} \leq \chi_f \left( U (r,m,k) \right)/ \frac{k}{r} \leq \frac{5}{4}\text{e}^2$.
\end{lemma}
 \begin{proof}
  The first part of the inequality follows from the fact that $\bar{G}_d$ is homomorphic to $U_d (r,m,k)$ as discussed before. It is easy to see that $U_d(r,m,k)$ is vertex transitive. Therefore, $\chi_f \left( U (r,m,k) \right) = \frac{\lvert V \left( U (r,m,k) \right)\rvert}{ \alpha \left( U (r,m,k)\right)}$ where $\alpha(.)$ is the independence number of the graph. To upper bound the ratio $\lvert V \rvert/ \alpha$, we provide a lower bound for the independence number by constructing a suitably large independent set. 

Consider $p$ elements from the set $[m]$ and denote the set by $Z$. Let $1 \leq p \leq m-k+1$. Pick one out of the $p$ elements in $Z$, say $z$. Now create vertices $(X,A)$ as follows: $z \in X$ and choose $k-1$ elements out of remaining $m-p$ elements. Then choose $r-1$ elements out of this and put in $X$ apart from $z$. The remaining $k-r$ elements form $A$ . Let $I_z$ denote the set of $(X,A)$ vertices created by the above method. Since $z \in X,~ \forall (X,A) \in I_z$, $I_z $ is an independent set by the definition of $U (r,m,k)$. We now argue that $\bigcup \limits_{x \in Z} I_x =I$ is an independent set. We need to argue that there cannot be an edge between $(X,A) \in I_x$ and $(Y,B) \in I_y$ for $x \neq y$. $x \in X$ but $x \notin B$ because when $I_y$ (and in particular $(Y,B) \in I_y$) is formed $x$ is not considered at all. Similarly, $y \in Y$ but $y \notin A$. Hence, $I$ is an independent set. $\lvert I_z\rvert = {m-p \choose k-1} {k-1 \choose r-1}$. 

     Hence , $\alpha \left ( U (r,m,k) \right) \geq \max \limits_{1 \leq p \leq m-k+1} p {m-p \choose k-1} {k-1 \choose r-1} $. Substituting the lower bound , we have the following:
\begin{align}
    \chi_f \left( U (r,m,k) \right)/ \frac{k}{r} \leq \frac { \lvert V \rvert } {\alpha \frac{k}{r}} & \leq \frac{{m \choose k} {k \choose r}}{ \max \limits_{1 \leq p \leq m-k+1} p {m-p \choose k-1} {k-1 \choose r-1} \frac{k}{r} }  \\ 
        \hfill                                                                                                                                   & = \frac{{m \choose k} {k \choose r}}{ \max \limits_{1 \leq p \leq m-k+1} p {m-p \choose k-1} {k \choose r}  }  \\
        \hfill                                                                                                                                  &  =\frac{{m \choose k}} {\max \limits_{1 \leq p \leq m-k+1} p {m-p \choose k-1} }
        \hfill                                                                                                                                  &\hspace{-50pt} = \frac{m {m-1 \choose k-1}} {k \max \limits_{1 \leq p \leq m-k+1} p {m-p \choose k-1}}  
\end{align}
  The last expression is identical to the ratio of $\chi_f$ to $\chi_\ell$ of $U(m,k)$ for the directed local chromatic number in (\ref{eq1}) in Theorem \ref{localbndtheorem}. We have seen the ratio is upper bounded by $\frac{5}{4}\text{e}^2$ when $k \geq 2$ in Theorem \ref{localbndtheorem} . Hence, we conclude our proof. 
\end{proof}

\begin{proof}[Proof of Theorem \ref{fraclocalbndthm}]
From the lemma above, $\chi_f(\bar{G}_d)/\left( \chi_\ell(\bar{G}_d,r)/r \right) \leq \frac{5}{4}\text{e}^2 $.  

Hence,  $ \chi_f(\bar{G}_d)/\left( \chi_{f\ell} (\bar{G}_d) \right) = \chi_f(\bar{G}_d)/\left( \inf \limits_{r} \chi_\ell(\bar{G}_d,r)/r \right) = \sup \limits_{r} \chi_f(\bar{G}_d)/\left(  \chi_\ell(\bar{G}_d,r)/r \right) \leq \frac{5}{4}\text{e}^2$.
\end{proof}
	
\section{Conclusion}
        We presented novel index coding upper bounds based on local and fractional local chromatic numbers for the case when user requests are distinct. We presented a problem instance where the additive gap between the new bounds and the bound based on fractional chromatic number is arbitrarily large. 
We also proved that the multiplicative gap between the new bounds and the bound based on fractional chromatic number is at most a constant. Studying local coloring concepts in the context of the general index coding problem with overlapping user requests is an interesting direction for future study.

\pagenumbering{arabic}
\bibliographystyle{IEEEtran}
\bibliography{fracchrom.bib}

	\end{document}

%% file: macros.tex
\long\def\comment#1{}

\newfont{\bbb}{msbm10 scaled 700}

\newfont{\bb}{msbm10 scaled 1100}






%% file: Fractionalchrom_v6.bbl
\begin{thebibliography}{10}
\providecommand{\url}[1]{#1}
\csname url@samestyle\endcsname
\providecommand{\newblock}{\relax}
\providecommand{\bibinfo}[2]{#2}
\providecommand{\BIBentrySTDinterwordspacing}{\spaceskip=0pt\relax}
\providecommand{\BIBentryALTinterwordstretchfactor}{4}
\providecommand{\BIBentryALTinterwordspacing}{\spaceskip=\fontdimen2\font plus
\BIBentryALTinterwordstretchfactor\fontdimen3\font minus
  \fontdimen4\font\relax}
\providecommand{\BIBforeignlanguage}[2]{{%
\expandafter\ifx\csname l@#1\endcsname\relax
\typeout{** WARNING: IEEEtran.bst: No hyphenation pattern has been}%
\typeout{** loaded for the language `#1'. Using the pattern for}%
\typeout{** the default language instead.}%
\else
\language=\csname l@#1\endcsname
\fi
#2}}
\providecommand{\BIBdecl}{\relax}
\BIBdecl

\bibitem{birk1998informed}
Y.~Birk and T.~Kol, ``Informed-source coding-on-demand (iscod) over broadcast
  channels,'' in \emph{INFOCOM'98. Seventeenth Annual Joint Conference of the
  IEEE Computer and Communications Societies. Proceedings. IEEE}, vol.~3.\hskip
  1em plus 0.5em minus 0.4em\relax IEEE, 1998, pp. 1257--1264.

\bibitem{bar2006index}
Z.~Bar-Yossef, Y.~Birk, T.~Jayram, and T.~Kol, ``Index coding with side
  information,'' in \emph{Foundations of Computer Science, 2006. FOCS'06. 47th
  Annual IEEE Symposium on}.\hskip 1em plus 0.5em minus 0.4em\relax IEEE, 2006,
  pp. 197--206.

\bibitem{alon2008broadcasting}
N.~Alon, E.~Lubetzky, U.~Stav, A.~Weinstein, and A.~Hassidim, ``Broadcasting
  with side information,'' in \emph{Foundations of Computer Science, 2008.
  FOCS'08. IEEE 49th Annual IEEE Symposium on}.\hskip 1em plus 0.5em minus
  0.4em\relax IEEE, 2008, pp. 823--832.

\bibitem{blasiak2010index}
A.~Blasiak, R.~Kleinberg, and E.~Lubetzky, ``Index coding via linear
  programming,'' \emph{arXiv preprint arXiv:1004.1379}, 2010.

\bibitem{maleki2012index}
H.~Maleki, V.~Cadambe, and S.~Jafar, ``Index coding: An interference alignment
  perspective,'' in \emph{Information Theory Proceedings (ISIT), 2012 IEEE
  International Symposium on}.\hskip 1em plus 0.5em minus 0.4em\relax IEEE,
  2012, pp. 2236--2240.

\bibitem{effros2012equivalence}
M.~Effros, S.~Rouayheb, and M.~Langberg, ``An equivalence between network
  coding and index coding,'' \emph{arXiv preprint arXiv:1211.6660}, 2012.

\bibitem{erdos1986coloring}
P.~Erd\H{o}s, Z.~F{\"u}redi, A.~Hajnal, P.~Komj{\'a}th, V.~R{\"o}dl, and
  A.~Seress, ``Coloring graphs with locally few colors,'' \emph{Discrete
  mathematics}, vol.~59, no.~1, pp. 21--34, 1986.

\bibitem{korner2005local}
J.~K{\"o}rner, C.~Pilotto, and G.~Simonyi, ``Local chromatic number and
  {S}perner capacity,'' \emph{Journal of Combinatorial Theory, Series B},
  vol.~95, no.~1, pp. 101--117, 2005.

\bibitem{korner1998zero}
J.~Korner and A.~Orlitsky, ``Zero-error information theory,'' \emph{Information
  Theory, IEEE Transactions on}, vol.~44, no.~6, pp. 2207--2229, 1998.

\bibitem{Simonyipersonal}
G.~Simonyi, G.~Tardos, and A.~Zsban, ``Personal communication.''

\bibitem{tehrani2012bipartite}
A.~Tehrani, A.~Dimakis, and M.~Neely, ``Bipartite index coding,'' in
  \emph{Information Theory Proceedings (ISIT), 2012 IEEE International
  Symposium on}.\hskip 1em plus 0.5em minus 0.4em\relax IEEE, 2012, pp.
  2246--2250.

\bibitem{simonyi2006local}
G.~Simonyi and G.~Tardos, ``Local chromatic number, {K}y {F}an's theorem, and
  circular colorings,'' \emph{Combinatorica}, vol.~26, no.~5, pp. 587--626,
  2006.

\bibitem{shokrollahi2006raptor}
A.~Shokrollahi, ``Raptor codes,'' \emph{Information Theory, IEEE Transactions
  on}, vol.~52, no.~6, pp. 2551--2567, 2006.

\bibitem{simonyi2010directed}
G.~Simonyi and G.~Tardos, ``On directed local chromatic number, shift graphs,
  and {B}orsuk-like graphs,'' \emph{Journal of Graph Theory}, vol.~66, no.~1,
  pp. 65--82, 2010.

\bibitem{godsil2001algebraic}
C.~Godsil, G.~Royle, and C.~Godsil, \emph{Algebraic graph theory}.\hskip 1em
  plus 0.5em minus 0.4em\relax Springer New York, 2001, vol.~8.

\end{thebibliography}
